\documentclass[11pt, reqno]{amsart}
\usepackage{amsmath, amsthm, a4, latexsym, amssymb}

\def\@rmrk#1#2{\refstepcounter
    {#1}\@ifnextchar[{\@yrmrk{#1}{#2}}{\@xrmrk{#1}{#2}}}

\makeatletter\@addtoreset{equation}{section}\makeatother

\newfont{\bfit}{cmbxti10 scaled 2000}
\newfont{\biggi}{cmr12 scaled 2000}
\newtheorem{step}{STEP}

\newcommand{\bes}{\begin{step}}
\newcommand{\es}{\end{step}}
 \newcommand{\essinf}{{\rm essinf}\,}
 \newcommand{\R}{\mathbb{R}}
 \newcommand{\Z}{\mathbb{Z}}

 \newcommand{\me}{\mathbb{E}}
 
 \renewcommand{\P}{\mathbb{P}}
 
 \newcommand{\skria}{{\mathcal A}}

 \newcommand{\skrid}{{\mathcal D}}

 \newcommand{\skrif}{{\mathcal F}}
 \newcommand{\skrig}{{\mathcal G}}

 \newcommand{\skrik}{{\mathcal K}}
 \newcommand{\skril}{{\mathcal L}}
 \newcommand{\skrim}{{\mathcal M}}

 \newcommand{\skrip}{{\mathcal P}}

 \newcommand{\skrit}{{\mathcal T}}

 \newcommand{\skrix}{{\mathcal X}}
 
 \newcommand{\skriz}{{\mathcal Z}}
 
 \newcommand{\sfrac}[2]{\mbox{$\frac{#1}{#2}$}}

\def\1{{\mathchoice {1\mskip-4mu\mathrm l}      % Blackboard bold 1
{1\mskip-4mu\mathrm l}
{1\mskip-4.5mu\mathrm l} {1\mskip-5mu\mathrm l}}}

\newcommand{\eq}{\begin{equation}}
\newcommand{\en}{\end{equation}}

\renewcommand{\subsection}{\secdef \subsct\sbsect}
\newcommand{\subsct}[2][default]{\refstepcounter{subsection}
\vspace{0.15cm}
{\flushleft\bf \arabic{section}.\arabic{subsection}~\bf #1  }
\nopagebreak\nopagebreak}
\newcommand{\sbsect}[1]{\vspace{0.1cm}\noindent
{\bf #1}\vspace{0.1cm}}

\newtheorem{theorem}{Theorem}[section]
\newtheorem{lemma}[theorem]{Lemma}
\newtheorem{cor}[theorem]{Corollary}

\newtheoremstyle{thm}{1.5ex}{1.5ex}{\itshape\rmfamily}{}
{\bfseries\rmfamily}{}{2ex}{}

\newtheoremstyle{rem}{1.3ex}{1.3ex}{\rmfamily}{}
{\itshape\rmfamily}{}{1.5ex}{}
\theoremstyle{rem}

\refstepcounter{subsection}

\def\thebibliography#1{\section*{reference}
  \list%
  {\arabic{enumi}.}%                          *** style of reference number ***
    {\settowidth\labelwidth{[#1]}\leftmargin\labelwidth
    \advance\leftmargin\labelsep
    \parsep0pt\itemsep0pt
    \usecounter{enumi}}
    \def\newblock{\hskip .11em plus .33em minus .07em}
    \sloppy                   % \clubpenalty4000\widowpenalty4000
    \sfcode`\.=1000\relax}

\begin{document}
%%%%%%%%%%%%%%%%%%%%%%%%%%%%%%%%%%%%%%%%%%%%%%%%%%%%%%%%%%%%%%%%%%%%%%%%%%%%%%%
\title[Lossy  version  of  AEP for Hierarchical Data Structures]
{\Large Lossy  Asymptotic Equipartition property for Hierarchical Data Structures}

\author[Kwabena Doku-Amponsah]{}

\maketitle
\thispagestyle{empty}
\vspace{-0.5cm}

\centerline{\sc{By Kwabena Doku-Amponsah}}
\renewcommand{\thefootnote}{}
\renewcommand{\thefootnote}{1}
\renewcommand{\thefootnote}{}
\footnote{\textit{Address:} Statistics Department, University of
Ghana, Box LG 115, Legon,Ghana.\,
\textit{E-mail:\,kdoku@ug.edu.gh}.}
\renewcommand{\thefootnote}{1}
%\centerline{\textit{Lecturer in Statistics}}%
\centerline{\textit{University of Ghana}}
%\centerline{\textit{E-mail:\,kdoku@ug.edu.gh}.}
%\centerline{\small Version: \jobname;  \version}\vspace{0.2cm}%

\begin{quote}{\small }{\bf Abstract.}
This  paper  presents a rate-distortion theory for  hierarchical  networked   data structures  modelled  as   tree-indexed multitype  process.  To  be  specific, this  paper gives a  generalized   Asymptotic Equipartition Property (AEP)   for  the  Process.  The general  methodology of proof  of  the  AEP  are  process level  large  deviation  principles   for   suitably  defined  \emph{empirical  measures}  for  muiltitype Galton-Watson  trees.
\end{quote}\vspace{0.5cm}

\section{Introduction}
Rate  distortion  theory (RDT)  play  crucial  role   in  approximate pattern-matching  in  information  theory.  It  provides  the mathematical foundations for lossy data compression; it  takes  care of  the problem of looking  for  the minimal number of bits per symbol, as measured by the rate $R,$ that should be transmitted  over a channel, so that the source (input signal) can be approximately decipher at the receiver (output signal) without exceeding a given distortion $D.$ The  RDT  is  mostly  centered  around  a  lossy  version of  the AEP, see example \cite{CT91}.\\

Several lossy versions  of  the  AEP  have  been  formulated  for  linear  data  sources  including  stationary ergodic random  fields  on  the $\Z^d,$  the  $d-$  dimensional  Lattice. See  example \cite{DK02}  and  the  reference  therein.  This lossy AEP   have  been  applied to  strengthened  versions  of Shannon's  direct  source coding and  universal  coding  theorems, characterize the performance of "mismatched"  code books in  lossy  data compression, analyse the  performance  pattern-matching algorithms  for lossy compression (including  the  Lempel-Ziv  schemes), determine the  first  order asymptotics  of  waiting times( with distortion) between  stationary process and  characterize the  best  achievable  rate  of  weighted  codebooks  as an  optimal sphere-covering  exponent.  See \cite{DK02}.  In  Doku-Amponsah\cite{DA10} an  AEP   has  been  found  for  hierarchical  structured  data. Such naturally tree-like data exists and  are usually encountered  in communication studies, demographic studies, biological population studies and the field of physics. Example, the age structure of a given population is best modelled by genealogical trees. The lossy  version  of the  AEP in  \cite{DA10}  is  yet  to  be developed.\\

In  this  paper  we  develop  a Lossy  AEP  for  hierarchical  data  structures  modelled  as  multitype  Galton-Watson  trees. To  be  specific  about  this methodology, we   use  LDP   for  the  empirical  offspring  measure  of  the critical,  irreducible  multitype  Galon-Watson trees,  see  \cite{DA06}, to  prove  an LDP  for  two  dimensional  multitype  Galton-Watson  trees. Using  this  LDP  together  with  the techniques  employed  by  Dembo and Kontoyiannis~\cite{DK02} for  the  random field on  $\Z^2$  we obtain the  proof of  the Lossy AEP for  the hierarchical  data  structures. \\

The  outline  of  the  paper  is  given  as  follows. Generalized AEP for Multitype Galton-Watson Process  section  contain the  main  result  of  the  paper, Theorem~\ref{AEP1}.  LDP for two-dimensional multitype galton-watson process  section  gives  processs level LDP's, Theorem~\ref{AEP2}  and \ref{AEP4},   which  form  the  bases of  the  proof the  main  result of  the  paper.   Proof of Theorem~\ref{AEP1}, ~\ref{AEP2} and ~\ref{AEP4} section  provides  the  proofs  of  all  Process  Level  LDP's  for  the  paper  and  hence  the  main  result  of  the  paper.

\section{Generalized  AEP for  Multitype Galton-Watson Process}
\subsection{Main Result}

Consider  two  multitype  Galton-Watson  processes   $X=\big\{ (X(v), C_X(v)):\, v\in V\big\}$  and    $Y=\big\{ (Y(v), C_Y(v)):\, v\in V\big\}$  which  take  values  in  $\skrit=\skrit(\skrix)$  and  $\hat{\skrit}=\hat{\skrit}(\skrix),$ resp.,  the  spaces  of  finite  trees on  $\skrix.$  We  equip  $\skrit(\skrix)$, $\hat{\skrit}(\skrix)$  with  their  Borel $\sigma$ fields $\skrif$  and  $\hat{\skrif}.$  Let  $\P_x$  and  $\P_y$  denote  the  probability measures  of the entire  processes $X$ and  $Y.$    By  $\skrix$  we  denote a  finite  alphabet  and  write  $\displaystyle\skrix_k^*=\bigcup_{n=0}^k\{n\} \times
\skrix^n,$  where  $k\in\Z^{+}.$ We always  assume  that  $X$  and  $Y$  are  independent  of  each  other.

Throughout the  rest  of  the  article  we  will  assume that  $X$  and  $Y$  are  irreducible, critical  multitype  Galton-Watson  processes. See example \cite{DMS03}. For  $n\ge 1$,  let  $P_n$  denote  the  marginal  distribution  of  $X$  given  $|V(T)|=n$  taking with  respect  to  $\P_x$ and  $Q_n$  denote  the  marginal distribution  $Y$  given $|V(T)|=n$  with  respect  to $\P_y.$  Let  $\rho:\skrix\times\skrix^{*}\times\skrix\times\skrix^{*}\to[0,\infty)$  be  an arbitrary  non-negative  function and  define  a  sequence of  single-letter  distortion  measures  $\rho^{(n)}:\skrit\times\hat{\skrit}\to[0,\infty),$  $n\ge 1$  by
$$\rho^{(n)}(x,y)=\frac{1}{n}\sum_{v\in V}\rho\Big(\skria_x(v),\,\skria_y(v)\Big),$$

where $\skria_x(v)=(x(v), c_x(v))$  and  $\skria_y(v)=(y(v), c_y(v)).$  Given  $d\ge 0$   and  $x\in\skrit$ ,  we  denote the  distortion-ball  of  radius  $d$  by   $$B(x,d)=\Big\{y\in\hat{\skrit}:\,\, \rho^{(n)}(x,y)\le d\Big\}.$$

Theorem~\ref{AEP1}  below  is  the  generalized  Shannon-McMillan-Breiman  Theorem  or Lossy Asymptotic Equipartition  Property for  the  hierarchical  data  structures. Define  the  matrix  $A:\skrix^2\times\skrix^2\to \R_+\bigcup\{0\}$  by   $$A[(a,\hat{a}),(b,\hat{b})]=\sum_{(c,\hat{c})\in{\skrix^{*}}^2}m(a,c)m(\hat{a},\hat{c})\skrik_x\{c \,|\,b\}\skrik_y\{\hat{c} \,|\,\hat{b}\}.$$

 By  $x\,\skrid\, p$  we  mean  $x$  has distribution  $p.$ For $\pi$  the  eigen vector   corresponding  to  the  largest  eigen  value  $1$  of  the  matrix  $A,$ we  write  $$d_{av}=\langle \log \langle e^{t\rho(\skria_X,\,\skria_Y)},\pi_1\otimes \skrik_x\rangle,\pi_2\otimes\skrik_y\rangle$$  and  assume $d_{min}^{(n)}=\me_{P_n}\big[\essinf_{Y\,\skrid\, Q_n}\rho^{(n)}(X,Y)\big]$  converges to  $D_{min}.$  For  $n>1,$   we  write  $$R_n(P_n, Q_n, d):=\inf_{V_n}\Big\{\frac{1}{n}H(V_n\,\|\,P_n\times Q_n):\,V_n\in \skrim(\skrit\times\hat{\skrit})\Big\}$$
and  write $$d_{min}^{\infty}:=\inf\Big\{d\ge 0:\,\sup_{n\ge 1}R_n(P_n, Q_n, d)<\infty\Big\}.$$

 We  call  $\nu\in\skrim[{(\skrix\times\skrix_k^*)}^2]$  with  marginals  $\nu_1$  and  $\nu_2$  respectively,   shift-invariant   if    $$\nu_{1,1}(a)=\sum_{(b,c)\skrix\times\skrix^{*}}m(a,c)\nu_1(b,c),  \,  \mbox{ $a\in\skrix$}$$  and
$$\nu_{2,1}(a)=\sum_{(b,c)\skrix\times\skrix^{*}}m(a,c)\nu_2(b,c),  \,  \mbox{ $a\in\skrix$,}$$

$m(a,c)$  is  the  multiplicity  of  the  symbol  $a$  in  $c.$  See \cite{DMS03}.  We  define  the  rate  function  $I_1:\skrim[{(\skrix\times\skrix_k^*)}^2]\to [0,\, \infty]$  by

\begin{equation}\label{AEP3}
\begin{aligned}
I_1(\nu)= \left\{ \begin{array}{ll}H\big(\nu\,\|\,\nu_{1,1}\otimes\skrik_x\times\nu_{2,1}\otimes\skrik_y), & \mbox{if $\nu$  is  shift-invariant,}\\

 \infty & \mbox{otherwise.}

\end{array}\right.
\end{aligned}
\end{equation}

\begin{theorem}\label{AEP1}
Suppose  $X$  and  $Y$  are  critical, weakly irreducible Multitype Galton-Watson  trees  with  transition  kernels  $\skrik_x$  and  $\skrik_y.$ Assume $\rho$ are  bounded  function.  Then,
\begin{itemize}

\item[(i)] with  $\P_x-$ probability $1,$ conditional  on  the  event $\big\{X=x,V(T)=n\big\}$  the  random  variables  $\Big\{ \rho^{(n)}(x,Y)\Big\}$ satisfy  an  LDP   with  deterministic,  convex  rate-function  $$I_{\rho}(z):=\inf_{\omega}\Big\{I_1(\omega):\, \langle\rho, \,\omega\rangle=z\Big\}.$$
    %with  $\P_x-$ probability $1.$
\item [(ii)] for  all $d\in(d_{min},\,d_{av})$,  except  possibly  at $d=d_{min}^{\infty}$
\begin{equation}\label{AEP11}
\lim_{n\to\infty}-\frac{1}{n}\log Q_n^{x}\Big(B(X,D)\Big)=R(\P^{x},\P^{y},d)\,\,\mbox {almost  surely,}\end{equation}
where
$R(p,q,D)=\inf_{\nu}H(\nu\,\|\,p\times q).$
\end{itemize}
\end{theorem}
\subsection{Application~~\cite{DA16}}

{\bf  Mutations  in  mitochondrial  DNA.} Mitochondria  are
organelles  in  cells  carrying  their  own  DNA.  Like  nuclear
DNA,  mtDNA  is  subject  to  mutations  which  may  take   the form
of  base  substitutions,  duplication  or  deletions.  The
population mtDNA  is  modelled  by  two-type  process  where  the
units  are  $1$  (normals)  and  $0$ (mutant),  and  the  links are
mother-child  relations.  A  normal  can  give  birth  to  either
all normals  or,  if   there  is  mutation,   normals  and mutants.
Suppose  the  latter  happens  with probability or  mutation  rate
$\alpha\in[0,\,1].$  Mutants can only give birth to mutants. A DNA molecule
may also die without reproducing. We  denote  by  $\emptyset$ the
 event absence of offspring. % Let the survival probabilities be $p\in\big[0,\sfrac{1}{(2-\alpha)}\big]$ and $q\in\big[0,\,\sfrac{1}{2}\big]$ for  normals  and  mutants respectively.
 Assume that the population is started from one normal
ancestor. Suppose the offspring kernel $\skrik$ is given by
$$\begin{aligned}
&\skrik\big\{(2,
a_1,a_2)\,|\,1\big\}=\Big(\frac{1}{2}\Big)\prod_{k=1}^{2}K_{\alpha}\{a_k\,|\,1\},\\
&\skrik\big\{(2,
a_1,a_2)\,|\,0\big\}=\Big(\frac{1}{2}\Big)\prod_{k=1}^{2}K_{\alpha}\{a_k\,|\,0\},\\
\end{aligned}$$

where $K_{\alpha}\{\emptyset\,|\,1\}=0,$
$K_{\alpha}\{0\,|\,1\}= \alpha,$
$K_{\alpha}\{1\,|\,1\}=(1-\alpha),$ $K_{\alpha}\{0\,|\,0\}=1$  and
$K_{\alpha}\{\emptyset\,|\,0\}=0.$  Note  that the  matrix  $A$  given  by
$$A=\begin{pmatrix}

1    &  0    &    0       & 0\\
\alpha & 1-\alpha &  0       & 0\\
\alpha  &0    &   1-\alpha & 0\\
\alpha^2 &\alpha(1-\alpha) &\alpha(1-\alpha) & (1-\alpha)^2
\end{pmatrix}$$

is  weakly  irreducible  $4\times4$  matrix,  see \cite{DMS03},  with   largest  eigen value  $1$  and  the  corresponding  eigen  vector  given  by
  $$ \pi=\begin{pmatrix}
  &\pi(0,0)& \\
  &\pi(0,1) &\\
  & \pi(1,0) &\\
  & \pi(1,1)&
  \end{pmatrix}=\begin{pmatrix}
  &\sfrac{1}{4} & \\
  &\sfrac{1}{4} & \\
  & \sfrac{1}{4} & \\
  & \sfrac{1}{4} &
  \end{pmatrix}$$

%Let $\beta(\pi)= \pi_1(1)\pi_2(1)+\pi_1(1)\pi_2(0)+\pi_1(0)\pi_2(0)$  and $\beta_2(\pi)=\pi_1(1)\pi_2(1)\alpha(1-\alpha)^{2}.$
 Therefore,  Theorem~\ref{AEP1}  hold  with  the  distortion-rate

 \begin{equation}\label{AEP11}
\begin{aligned}
R(P,Q,D)=\left\{ \begin{array}{ll} 0, & \mbox{ if $D\ge  \sfrac {3}{4}(1-\alpha)+\sfrac{1}{4}\alpha(1-\alpha)^3$,}\\

 \infty & \mbox{otherwise.}

\end{array}\right.
\end{aligned}
\end{equation}

\section{LDP  for  two-dimensional  multitype  galton-watson process}

Given  a  probability  measure  $\mu:\skrix\times\hat{\skrix}\to [0,1]$  and  transition  kernel  $\skrik$  we  define  the  two-dimensional  multype  Galton Watson  tree  as  follow:
\begin{itemize}
\item  Assign  the  root $\eta$  type  $(X(\eta),Y(\eta))$  independently  according  to  $\mu.$
\item  Give  any  vertex $v$  with  type  $(a,b)$   offspring  types  and  number  of  springs  $C_{X,Y}(v)$  independent  everything according  to $$\skrik\Big\{(C_{X,Y}(v)=(c_a,c_b)\|(a,b)\Big\}=\skrik_x\Big\{c_a\|a\Big\}\skrik_y\Big\{c_b\|b\Big\}$$
\end{itemize}
We  define  the  process-level  empirical  measure  $\skril_n$  induced  by  $X$  and  $Y$  on  $\skrit\times\hat{\skrit}$   by
$$\skril_n(a_x,a_y)=\frac{1}{n}\sum_{v\in V}\delta_{\big(\skria_X(v),\,\skria_Y(v)\big)}(a_x,a_y), \, \mbox{ for $(a_x,a_y)\in\skrim[{(\skrix\times\skrix_k^*)}^2].$ }$$
%where  $\phi(a_x(v),\skria_y(v))=\big((x(v),y(v)),\, c_{x,y}(v)\big).$
Note  that  we  have
$$\begin{aligned}
\skril_n\otimes \phi^{-1}\big((x(v),y(v)),\, c_{x,y}(v)\big)&=\frac{1}{n}\sum_{v\in V}\delta_{\big(\skria_X(v),\,\skria_Y(v)\big)}\otimes \phi^{-1}\big((x(v),y(v)),\, c_{x,y}(v)\big)\\
&=\frac{1}{n}\sum_{v\in V}\delta_{\big((X(v),Y(v)),\, C_{X,Y}(v)\big)}\big((x(v),y(v)),\, c_{x,y}(v)\big)\\
&:=\tilde{\skril}_n\big((x(v),y(v)),\, c_{x,y}(v)\big),
\end{aligned}$$

where  $\phi(\skria_x(v),\skria_y(v))=\big((x(v),y(v)),\, c_{x,y}(v)\big).$  The  next  Theorem  which is  the  LDP for  $\skril_n$  of  the  process  $X,Y$  is  the  main ingredient  in  the  proof  of  the  Lossy  AEP.
\begin{theorem}\label{AEP2}
The  sequence  of empirical  measures  $\skril_n$ satisfies a  large  deviation  principle  in  the  space  of  probability  measures  on   $(\skrix\times{\skrix_k^*})^2$  equipped  with  the  topology of  weak  convergence,  with  convex,  good rate-function  $I_1.$

\end{theorem}

The  proof  of  Theorem\ref{AEP2}  above   is dependent  on  the LDP  for  $\tilde{\skril}_n$  given  below:

\begin{theorem}\label{AEP4}
The  sequence  of empirical  measures  $\tilde{\skril}_n$ satisfies a  large  deviation  principle  in  the  space  of  probability  measures  on   $\skrix^2\times{\skrix^{*}}^2$  equipped  with  the  topology of  weak  convergence,  with  convex,  good rate-function

\begin{equation}\label{AEP5}
\begin{aligned}
I_2(\omega)= \left\{ \begin{array}{ll}H\big(\omega\,\|\,\omega_1\otimes\skrik_x\times\skrik_y), & \mbox{if
$\omega$  is  shift-invariant,}\\
\infty & \mbox{otherwise,}

\end{array}\right.
\end{aligned}
\end{equation}

where  $\omega_1\otimes\skrik_x\times\skrik_y\big((a,b),(c_{a},c_{b})\big)=\omega_1(a,b)\skrik_x\{c_a \,|\,a\}\skrik_y\{c_b \,|\,b\}.$
\end{theorem}

\section{Proof  of  Theorem~\ref{AEP1}, \ref{AEP2} and~\ref{AEP4}}

\subsection{Proof  of  Theorem~\ref{AEP4}}

\begin{cor}[\cite{DA16}]\label{AEP5}
Let $Z$ be a weekly irreducible, critical multitype
Galton-Watson tree with an  offspring law  $Q$ whose second moment
is finite, conditioned to have exactly $n$ vertices. Then, for
$n\to\infty$, the empirical offspring measure~$\skrim_Z$ satisfies an LDP in $\skrip\big[\skriz\times\skriz^*\big]$ with
speed $n$ and the convex, good rate function
\begin{equation}\label{equ-ratecor}
\Phi_Q(\varpi)=\left\{ \begin{array}{ll} H(\varpi \,\|\, \varpi_1\otimes Q) & \,
\mbox{ if \,$\varpi$ is  weak shift-invariant,}\\
\infty & \mbox{ otherwise.}
\end{array} \right.\end{equation}
\end{cor}

The  proof  of  Theorem~\ref{AEP4} follows  from  Corollary~\ref{AEP5} by   the  contraction  principle,  see \cite{DZ98}  applied  to  the linear  mapping  given  by  $\skrig(\skrim_{(x,y)})=\tilde{\skril}_{n}$,\, $z=(x,y).$ The  rate  function  governing  this  LDP  is  given  by  $$ I_2(\omega)=\Big\{\Phi_{Q}(\varpi):\,\varpi=\omega,\, Q=\skrik\times\skrik\Big\}.$$

We  obtain the  form  of  the  rate  function  $I_2$  in  \eqref{equ-ratecor}  if  we  note  that  $\skrip(\skrix^2\times{\skrix_k^*}^2)$  where  $\skriz=\skrix^2$  and  $\skriz^*={\skrix_k^*}^2.$

\subsection{Proof  of  Theorem~\ref{AEP2}}
\begin{lemma}\label{AEP6}
$\tilde{\skril}\otimes\phi$   obeys  an  LDP   on  the  space   $\skrip\big[\phi(\skrix^2\times{\skrix_k^*}^2)]$ with  good  rate  function $I_2,$
where  $\phi^{-1}\big((a,b),(c_a,c_b)\big)=(\skria_a,\skria_b)$  and  $a_z=(z,c_z).$
\end{lemma}
\begin{proof}
Let  $\Gamma\in\skrip\big((\skrix\times{\skrix_k^*})^2\big)$  and  write  $ \Gamma_{\phi}=\big\{\omega:\, \omega\otimes\phi\in \Gamma\big\}.$  Note  that  if  $A$  is  closed (open)  then  $\Gamma_{\phi}$  is  closed (open) since  $\rho$  is  bounded.  Now  suppose    $F$  is  closed  subset  of   $\skrip[{(\skrix\times\skrix_k^*)}^2]$   then  we  have
$$\lim_{n\to\infty}\sfrac{1}{n}\log\P\big\{ \tilde{\skril}\otimes \phi\in F\big\}=\lim_{n\to\infty}\sfrac{1}{n}\log\P\big \{ \tilde{\skril}\in F_{\phi}\big\}\le -\inf_{\omega\in F_{\phi}}I_2(\omega)=-\inf_{\omega\otimes \phi\in F} I_2(\omega\otimes \phi)$$

Suppose  G  is  open subset  of   $\skrip[{(\skrix\times\skrix_k^*)}^2]$   the  we  have also  that  $$\lim_{n\to\infty}\sfrac{1}{n}\log\P\big\{ \tilde{\skril}\otimes \phi\in G\big\}=\lim_{n\to\infty}\sfrac{1}{n}\log\P\big \{ \tilde{\skril}\in G_{\phi}\big\}\ge -\inf_{\omega\in G_{\phi}}I_2(\omega)=-\inf_{\omega\otimes \phi\in G}I_2(\omega\otimes \phi)$$

\end{proof}

  By  Lemma~\ref{AEP6} and  the  contraction  principle  applied  to  the  linear  mapping  $\tilde{\skril}(\phi)(\skria_x,\skria_y)=\skril(\skria_x,\skria_y),$  we  have  that  $\skril$   obeys  a  LDP  on  the  space  $\skrip\big[\phi(\skrix^2\times{\skrix_k^*}^2)]=\skrip\big[(\skrix\times{\skrix_k^*})^2\big]$ with  rate  function  $$I_1(\nu)=\Big\{I_2(\omega\otimes \phi):\, \omega\otimes \phi=\nu \Big\}=\begin{aligned} \left\{ \begin{array}{ll}H\big(\nu\,\|\,\nu_{1,1}\otimes\skrik\times\nu_{2,1}\otimes\skrik), & \mbox{if $\nu$  is  shift-invariant,}\\
\infty & \mbox{otherwise.}

\end{array}\right.
\end{aligned}$$

%\pagebreak
\subsection{Proof  of  Theorem~~\ref{AEP1}}

(i)    Notice  $\displaystyle\rho^{(n)}(X,Y)=\langle\rho, \,\skril_n^{(X,Y)}\rangle$  and  if  $\Gamma$  is open (closed)   subset  of  $\skrim((\skrix\times{\skrix_k^*})^2)$ then  $$ \Gamma_{\rho}:=\big\{ \omega: \langle\rho, \,\omega\rangle\in \Gamma\big\}$$ is  also open (closed) set  since  $\rho$  is  bounded function.

$$\begin{aligned}
-\inf_{z\in In(\Gamma)}I_{\rho}(z)&=-\inf_{\omega\in\ln(\Gamma_{\rho})}I_1(\omega)\\
&\le\liminf_{n\to\infty}\sfrac{1}{n}\log\P \Big\{\rho^{(n)}(X,Y)\in\Gamma\big |X=x,V(T)=n\Big\}\\
 &\le\lim_{n\to\infty}\sfrac{1}{n}\log\P \Big\{\rho^{(n)}(X,Y)\in\Gamma\big |X=x,V(T)=n\Big\}\\ &\le\limsup_{n\to\infty}\sfrac{1}{n}\log\P\Big\{\rho^{(n)}(X,Y)\in\Gamma\big |X=x,V(T)=n\Big\}\le -\inf_{\omega\in cl(\Gamma_{\rho})}I_1(\omega)=-\inf_{z\in cl(\Gamma)}I_{\rho}(z).
 \end{aligned}$$
% where  $I_{\rho}(z):=\inf_{\omega}\Big\{I_1(\omega):\, \langle\rho, \,\omega\rangle=z\Big\}.$

(ii)  Observe  that  $\rho$ are  bounded, therefore  by  Varadhan's  Lemma and  convex duality, we  have
$$R( \P^x, \P^y, d)=\sup_{t\in\R}[td-\Lambda_{\infty}(t)]=\Lambda_{\infty}^{*}(d)$$
where
$$\Lambda_{\infty}^{*}(t):=\lim_{n\to\infty}\sfrac{1}{n}\log \int e^{nt\Big\langle\rho, \,\skril_n^{(X,Y)}\Big\rangle}dQ_n(y)$$
exits  for  $\P$ almost  everywhere  $x.$  Using  bounded  convergence,  we  can  show  that  $$\Lambda_{\infty}(t)=\lim_{n\to\infty}\Lambda_n(t)
:=\lim_{n\to\infty}\sfrac{1}{n}\int \Big[\log \int e^{nt\Big\langle\rho, \,\skril_n^{(X,Y)}\Big\rangle}dQ_n(y)\Big]dP_n(x).$$

 Define  the  matrix  $A:\skrix^2\times\skrix^2\to \R_+\bigcup\{0\}$  by   $$A[(a,\hat{a}),(b,\hat{b})]=\sum_{(c,\hat{c})\in{\skrix^{*}}^2}m(a,c)m(\hat{a},\hat{c})\skrik_x\{c \,|\,b\}\skrik_y\{\hat{c} \,|\,\hat{b}\}.$$ As   $X,Y$  is  critical  and irreducible  the  matrix  $A$ is  irredicible  and  the  largest  eigen value  is $1.$   Therefore there  exists a unique Perron-Frobenius  eigen vector $\pi,$ ( normalized to  probability vector) corresponding  to this  largest  eigen  value   (of  the  matrix $A$), see \cite[Lemma~3.18]{DA10}, such  that
$$\sfrac{1}{n}\Lambda(nt)=\frac{1}{n}\sum_{j=1}^{n}\log\me_{Q_n}\big(e^{t\rho(\skria_x(j),\skria_y(j)}\big)\to\langle \log \langle e^{t\rho(\skria_X,\skria_Y)},\pi_1\otimes \skrik_x\rangle,\pi_2\otimes\skrik_y\rangle=d_{av},$$
where  $\pi_1$  and  $\pi_2$  are  the  first  and  second  marginals  of  $\pi,$ respectively. Recall  that by  $x\,\skrid\,\, p$ we mean  $x$ distributed  as  $P.$  Also let  $$D_{min}^{(n)}:=\lim_{t\downarrow-\infty}\sfrac{\Lambda_n(t)}{t}$$
so  that  $\Lambda_n^{*}(d)=\infty$ for  $d< d_{min}^{(n)}$,  while   $\Lambda_n^{*}(D)<\infty$ for  $d>d_{min}^{(n)}.$  Observe  that  for  $n<\infty$  we  have   $D_{min}^{(n)}(d)=\me_{P_n}\big[\essinf_{Y\,\skrid\, Q_n}\rho^{(n)}(X,Y)\big],$  which  converges to  $d_{min}.$
 Using  similar  arguments  as  \cite[Proposition~2]{DK02}  we  obtain  $$ R_n(P_n,Q_n,d)=\sup_{t\in\R}\big(td-\Lambda_n(t)\big):=\Lambda_{n}^{*}(d)$$

Now  we   observe  from \cite[Page 41]{DK02}  that  the  converge  of  $\Lambda_{n}^{*}(\cdot)\to\Lambda_{\infty}(\cdot)$  is  uniform on  compact  subsets  of  $\R.$  Moreover, $\Lambda_{n}$    convex,  continuous  functions    converge  informally  to  $\Lambda_{\infty}$  and  hence  we  can invoke  \cite[Theorem 5]{Sce48}  to  obtain

 $$\Lambda_{n}^{*}(d)=\lim_{\delta\to 0}\limsup_{n\to\infty}\inf_{|\hat{d}-d|<\delta}\Lambda_{n}^{*}(\hat{d}).$$

  Using similar  arguments as \cite[Page 41]{DK02} in  the lines  after  equation  (64)  we  have  \eqref{AEP11}  which  completes  the  proof.

%%%%%%%%%%%%%%%%%%%%%%%%%%%%%%%%%%%%%%%%%%%%%%%%%%%%%%%%%%%%%%%%%%%%%%%%%%%%%%%

%{Kwabena Doku-Amponsah, University of Ghana, Department of
%Statistics, P.O. Box LG 115,\\ Legon-Accra, Ghana. E-mail:
%\texttt{kdoku@ug.edu.gh}.}

\end{document}